\documentclass[a4paper, UKenglish, cleveref, autoref, thm-restate]{lipics-v2021}

\nolinenumbers 

\usepackage[T1]{fontenc}

\usepackage{graphicx}
\usepackage{tikz}
\usepackage{tikz-qtree}
\usetikzlibrary{automata,positioning}
\usetikzlibrary {graphs}
\usetikzlibrary {arrows.meta}
\usetikzlibrary{decorations.pathreplacing,calc}
\usepackage{bm}

\tikzset{%
  middle dotted line/.style={
    decoration={show path construction, 
      lineto code={
          \draw[#1] (\tikzinputsegmentfirst) --($(\tikzinputsegmentfirst)!.3333!(\tikzinputsegmentlast)$);,
          \draw[dotted,#1] ($(\tikzinputsegmentfirst)!.3333!(\tikzinputsegmentlast)$)--($(\tikzinputsegmentfirst)!.6666!(\tikzinputsegmentlast)$);,
          \draw[#1, ->] ($(\tikzinputsegmentfirst)!.6666!(\tikzinputsegmentlast)$)--(\tikzinputsegmentlast);,
      }
    },
    decorate
  },
}
\tikzset{smalln/.style={state, minimum size=0pt, inner sep=1pt, fill=white, draw=black}}

\usepackage{amsmath,amssymb,stmaryrd}
\newtheorem{innercustomthm}{Theorem}
\newenvironment{customthm}[1]
  {\renewcommand\theinnercustomthm{#1}\innercustomthm}
  {\endinnercustomthm}

\usepackage{chngcntr}
\usepackage{apptools}
\AtAppendix{\counterwithin{theorem}{section}}
\AtAppendix{\counterwithin{lemma}{section}}

\usepackage{mathrsfs}
\usepackage{bussproofs} \EnableBpAbbreviations

\usepackage{arydshln}

\usepackage{cprotect}

\def\rddots#1{\cdot^{\cdot^{\cdot^{#1}}}}

\usepackage{listings}

\newcommand{\set}[1]{\{ #1 \}}

\newcommand{\var}[1]{\texttt{\symbol{"5C}}#1}
\newcommand{\tuple}[1]{\langle #1 \rangle}

\newcommand{\nlog}{\textbf{\textsc{{NLOG}}}}
\newcommand{\np}{\textbf{NP}}
\newcommand{\pspace}{\textbf{PSPACE}}
\newcommand{\rewbl}{\textsc{REWBLk}}
\newcommand{\rewb}{REWB}
\newcommand{\prewb}{\textsc{REWB}(+)}
\newcommand{\nrewb}{\textsc{REWB}(\text{$-$})}
\newcommand{\indexed}{\textbf{IL}}

\newcommand{\ilog}[1]{\lceil \log #1 \rceil}
\newcommand{\rewblang}{\textbf{REWB}}
\newcommand{\rewbllang}{\textbf{\textsc{REWBLk}}}
\newcommand{\oraclem}{\textbf{OLOG}}

\usepackage[d]{esvect}

\newcommand{\depth}{\mathit{depth}}

\newcommand{\sem}[1]{\llbracket #1 \rrbracket}

\DeclareMathSymbol{\Gamma}{\mathalpha}{operators}{0}
\DeclareMathSymbol{\Sigma}{\mathalpha}{operators}{6}

\title{Regular Expressions with Backreferences and Lookaheads Capture NLOG}

\author{Yuya Uezato}{CyberAgent Inc., Japan \and National Institute of Informatics, Japan}{uezato\_yuya@cyberagent.co.jp}{https://orcid.org/0009-0005-8834-010X}{}
\authorrunning{Y. Uezato}
\Copyright{Yuya Uezato}


\begin{document}

\maketitle

\begin{abstract}
Backreferences and lookaheads are vital features to make classical regular expressions (REGEX) practical.
Although these features have been widely used, understanding of the unrestricted combination of them has been limited.
Practically, most likely no implementation fully supports them.
Theoretically, while some studies have addressed these features separately, few have dared to combine them.
In those few studies, it has been made clear that the amalgamation of these features renders REGEX significantly expressive.
However, no acceptable expressivity bound for \rewbl{}---REGEX with backreferences and lookaheads---has been established.
We elucidate this by establishing that \rewbl{} coincides with \nlog, the class of languages accepted by log-space nondeterministic Turing machines (NTMs).
In translating REWBLk to log-space NTMs, negative lookaheads are the most challenging part since it essentially requires complementing log-space NTMs in nondeterministic log-space.
To address this problem, we revisit Immerman--Szelepcs{\'e}nyi theorem.
In addition, we employ log-space nested-oracles NTMs to naturally handle nested lookaheads of \rewbl.
Utilizing such oracle machines, we also present the new result that the membership problem of \rewbl{} is \pspace-complete.
\end{abstract}

\section{Introduction}\label{sec:introduction}

Backreferences and lookaheads are practical extensions for classical regular expressions (REGEX).
REGEX with backreferences---\emph{REWB}---can represent the non context-free language $L = \set{ w \# w : w \in \set{a, b}^*}$  with the REWB expression $E = ((a+b)^*)_x\ \#\ \var{x}$.
Roughly telling about this expression,
we save a substring matched with $(a+b)^*$ into the variable $x$, and later, we refer back to the matched string using $\var{x}$; therefore, $E$ represents $L$.
REWB is a classical calculus~\cite{Aho:1990}, and there are some results:
\begin{enumerate}
\item Schmid showed that  \rewblang, the class of languages accepted by REWB, is contained in \nlog, the class of languages accepted by log-space nondeterministic Turing machines (NTMs)~\cite[Lemma~18]{Schmid:2013}.
\item Recently, Nogami and Terauchi showed that \rewblang{} is contained in the class of indexed languages, \indexed{} (an extension of context-free languages~\cite{Aho:1968})~\cite{Nogami:2023}.
\item The membership problem of REWB---for an input REWB expression $E$ and an input word w, deciding if $E$ accepts $w$---is \np-complete.
\end{enumerate}

On REGEX with lookaheads, there are two kinds of lookaheads.
A \emph{positive} lookahead $?(E) F$ checks if the rest of the input can be matched by $E$ \emph{without} consuming the input. If so, we continue to $F$ in the usual manner, i.e., with consuming the input.
A \emph{negative} lookahead $!(E) G$ checks if the rest of the input \emph{cannot} be matched by $E$ without consuming the input. If so, we continue to $G$ in the usual manner.
For example, the expression $?(E)F + !(E)G$ can be read as $\texttt{if}\ E\ \texttt{then}\ F\ \texttt{else}\ G$ and is helpful in writing practical applications.
Although lookaheads are useful, they do not alter the REGEX expressiveness.
This fact immediately follows from the result of alternating Turing machines that the language class of alternating finite automata corresponds to that of usual finite automata~\cite{Chandra:1981}.

Now, it is a natural question: \emph{How expressive are REWB with lookaheads?}

The class REWB with lookaheads---\textsc{REWBLk}---was studied by Chida and Terauchi~\cite{Chida:2022,Chida:2023}.
They have shown
(a) \rewbllang, the class of languages accepted by \rewbl, is closed under intersection and complement (it is immediately shown using lookaheads);
(b) REWB is a proper subclass of \prewb{}---REWB with positive lookaheads---and \nrewb{}---REWB with negative lookaheads;
(c) the language emptiness problem of \nrewb{} is undecidable.
They also developed a new class of automata called \emph{positive lookaheads memory automata (PLMFA)} and proved that \prewb{} equals PLMFA.

However, the following two key questions remain unresolved: \\
(i) Which known language classes are related to \rewbllang? \\
(ii) What is the computational complexity of the membership problem of \rewbl?

We solve these questions by presenting the following tight results:
\begin{enumerate}
\item[(I)] $\rewbllang = \nlog$, the class of languages accepted by \emph{log-space} \emph{nondeterministic} Turing machines (NTMs).
We also show that \rewblang{} contains an \nlog-complete language.
\item[(II)] The membership problem of \rewbl{} is \pspace-complete.
\end{enumerate}

Together with existing results, our results are summarized in the following table:\\
\begin{tabular}{c|c|c}
& Language Class  & Membership Problem  \\\hline
REGEX+lookaheads & = \textbf{Regular}~\cite{Chandra:1981} & \textbf{P}-complete~\cite{Jiang:1991} \\\hline
REWB & \begin{tabular}{cc}
$\subseteq \nlog$~\cite{Schmid:2013}, incomparable to \textbf{CFL}~\cite{Berglund:2023,Campeanu:2003}, \\\
$\subseteq \textbf{IL}$~\cite{Nogami:2023}, $\ni \nlog$-complete language (I)
\end{tabular} & \textbf{NP}-complete~\cite{Aho:1990} \\\hline
\rewbl & $= \nlog$ (I) & \pspace-complete (II) \\\hline
\end{tabular}

\subsection{Difficulty in Translating \rewbl{} to Log-space NTMs}\label{sec:intro:challenge}

To investigate a hard part of translation from \rewbl{} to log-space NTMs, let us consider the following language, which is a well-known \nlog-complete language:
\[
L_{\text{reach}} = \set{ s\,\#\,x_1 \to y_1 \# \cdots \# x_n \to y_n \#\,t : s, t, x_i, y_i \in V^*, \text{and there is a path from $s$ to $t$} }
\]
where the part $x_1\to y_1\#\cdots\#x_n\to y_n\#$ means the directed graph with direct edges $x_1 \to y_1$, $x_2 \to y_2$, and so on.
The following \rewbl{} expression $E_{\text{reach}}$ recognizes $L_{\text{reach}}$:
\[
E_{\text{reach}} = (V^*)_{\textsc{Cur}}\,\#\ \bigl(?(\Sigma^*\,\#\,\var{\textsc{Cur}}\to(V^*)_{\textsc{Cur}}\,\#)\bigr)^*\,\Sigma^*\,\#\,\var{\textsc{Cur}}
\]
where $\Sigma = V \cup \set{ \#, \to }$.
It first captures $s$ into the variable \textsc{Cur} and walks on graphs while repeatedly evaluating the part $?(\cdots)$.
Each evaluation makes a nondeterministic one-step move on the graph. The part $\Sigma^*\,\#\,\var{\textsc{Cur}}$ checks if we reach the goal $t$.

It is easy to structurally translate $E_{\text{reach}}$ to a log-space NTM $M$ such that $L(M) = L(E_{\text{reach}})$.
Now, using $M$, let us translate the if-then-else expression $?(E_{\text{reach}})F\,+\,!(E_{\text{reach}})G$ for some expressions $F$ and $G$ to a log-space NTM $N$.
Let $w$ be an input word $s \# \mathit{edges} \# t$.
For the part $?(E_{\text{reach}})F$, we can directly run $M$ in $N$, and if we eventually reach $t$ by running $M$ from $s$, then we continue to $F$ in $N$.

However, for the other part $!(E_{\text{reach}})G$, we encounter problems:
\begin{description}
\item[\textit{\textbf{(A)}}] We need to check if all possible walks of $M$ starting from $s$ do not reach $t$.
\item[\textit{\textbf{(B)}}] Walking paths starting from $s$ and aiming for $t$ become infinitely long, and thus there are infinitely many walks (branching) to be checked.
\end{description}
Therefore, we cannot run $M$ directly in $N$ to handle the negative lookahead $!(E_{\text{reach}})$.

\paragraph*{Our Idea: Immerman--Szelepcs{\'e}nyi Theorem \& Log-space Nested-oracles NTMs}
To address the above problems, we leverage Immerman--Szelepcs{\'e}nyi (I-S) theorem~\cite{Immerman:1988,Szelepcsenyi:1988}.
This theorem states that the class \nlog{} is closed under complement; i.e., there exists a log-space NTM $\overline{M}$ such that $L(\overline{M}) = \Sigma^* \setminus L(M)$.
Therefore, in our machine $N$, we run $\overline{M}$ for the part $!(E_{\text{reach}})G$: If $\overline{M}$ eventually accepts $w$, then we proceed to  simulate $G$.

On the other hand, \rewbl{} permits nested lookaheads, such as $?(\cdots !(\cdots ?(\cdots) \cdots) \cdots)$.
To handle them naturally,
we employ log-space NTMs with \emph{nested oracles}~\cite{Immerman:1988,Schoning:1988,Lange:1987}.
These machines can easily simulate \rewbl{} and are translated to log-space NTMs by the (I-S) theorem; so, they lead us to $\rewbllang = \nlog$.
Moreover, oracle machines are crucial to showing that the membership problem of \rewbl{} is \pspace-complete.
To this end, we also give the new result that the membership problem of such machines is in \pspace.

\paragraph*{Structure of Paper}
The rest of the paper is structured as follows.
Section~\ref{sec:relwork} discusses related work.
Section~\ref{sec:rewbl} reviews \rewbl{} and demonstrates that $\rewblang$ already contains an \nlog-complete language.
Section~\ref{sec:expressiveness} illustrates the expressiveness of \rewbl{}:
(1) $\nlog \subseteq \rewbllang$;
(2) the membership problem of \rewbl{} is \pspace-hard;
(3) \prewb{} and \nrewb{} represent languages $\not\in$ \indexed{}; and
(4) the emptiness problems of \prewb{} and \nrewb{} are undecidable even if $\Sigma = \set{ a }$.
Section~\ref{sec:nested-oracle-logspace-ntm} reviews log-space nested-oracles NTMs and their language class (= \nlog{}),
and shows our new result: their membership problem is in \pspace.
Section~\ref{sec:rewbl-to-nlog} establishes  $\rewbllang \subseteq \nlog$ and that the membership problem of \rewbl{} is in \pspace.
Section~\ref{sec:conclusion} concludes this paper by giving open problems.

\section{Related Work}\label{sec:relwork}

As discussed in the Introduction, Chida and Terauchi have formalized \rewbl{} and its semantics~\cite{Chida:2022,Chida:2023}.
To our knowledge, their study is the first theoretical exploration into the simultaneous treatment of backreferences and lookaheads.
Surprisingly, there has been no prior theoretical research on the topic despite their longstanding and widespread use.
They introduced PLMFA (positive lookahead MFA) by expanding MFA (memory finite automata), which Schmid presented for studying REWB in~\cite{Schmid:2016}.
One of their main result is the equivalence of PLMFA to \prewb{}, established through translations between PLMFA and \rewbl.
Nevertheless, they did not address (i) a relationship between \rewbl{} and existing known language classes; (ii) the complexity of the membership problem of \rewbl.
In contrast, we show that \rewbl{} captures \nlog{} and the membership problem of \rewbl{} is \pspace-complete.

As highlighted in the Introduction, for REWB, Schmid showed $\rewblang \subseteq \nlog$~\cite{Schmid:2013} and Nogami and Terauchi showed $\rewblang \subseteq \indexed$~\cite{Nogami:2023}.
Schmid also introduced a class MFA and showed that MFA corresponds to REWB~\cite{Schmid:2016}.
About the relationship between $\nlog$ and $\indexed$, it is worthy noting that:
\begin{itemize}
\item $\nlog \not\subseteq \indexed$. This is because that the language $L_{2^{2^n}} = \set{ a^{2^{2^n}} : n \in \mathbb{N} }$ clearly belongs to $\nlog$; however, $L_{2^{2^n}} \notin \indexed$, which is shown by pumping lemmas for indexed languages~\cite{Hayashi:1973,Gilman:1996}.
\item $\indexed \not\subseteq \nlog$ unless $\nlog = \np$. This is because that \indexed{} can represent the language $L_{\text{3SAT}}$, whose words are true 3-SAT formulas~\cite{Rounds:1973}; however, $L_{\text{3SAT}} \notin \nlog$ unless $\nlog = \np$.
\end{itemize}
In the present paper, we take their result further and show that $\rewbllang = \nlog$ and that $\rewblang$ already contains an \nlog-complete language.

We also refer to modern REGEX engines that (partially) support both backreferences and lookaheads.
Several programming languages (for example, Perl, Python, PHP, Ruby, and JavaScript) and .NET framework support these features.
However, their support is limited in both syntax and semantics.
First, expressions like $( \var{x} \var{x} )_{x}$ and $(\var{x} + \var{x})_{x}$ are rejected by most implementations because the variable $x$ appears more than once in single captures.
Next, in most implementations, the expression $F = (?((\var{x}a)_x))^* \var{x}$, which represents $\set{ \epsilon, a, a^2, \ldots }$, does not match with $a^2$ and $a^3$, so on.
It is due to conservative loop-detecting semantics. Such a semantics is standardized in ECMAScript~\cite{ecmascript}.
This semantics works for $F$ as follows.
First, it unfolds the Kleene-* of  $(?((\var{x}a)_x))^*$ as $(\epsilon\ +\ \underline{?((\var{x}a)_x)} (?((\var{x}a)_x))^*)$.
Next, it enters the underline part and updates the variable $x$ as $\epsilon \mapsto a$ without consuming any input characters.
Then, we try to evaluate $(?((\var{x}a)_x))^*$ \emph{again} at the same input position.
At this point, many REGEX engines think that we enter an infinite loop and so stop unfolding the {Kleene-*}.
Consequently, $F$ only matches with $\epsilon$ (without loop unfolding) and $a$ (with a single loop unfolding).

We can rephrase this situation as follows:
(1) the amalgamation of lookaheads with variables induces side effects without consuming any characters;
(2) however, the loop-detecting semantics overlooks such side effects and changes behaviors from the naive semantics.
On the other hand, such conservative semantics work well for REGEX, REGEX with lookaheads, or REWB since they do not induce such side effects.

This paper presents a translation between {\rewbl}s and log-space NTMs, and it enables to develop REGEX engines that fully support backreferences and lookaheads. 
Especially, such engines run in \emph{polynomial time} (for a fixed expression) since $\nlog \subseteq \textbf{P}$.

\section{Preliminaries: \rewbl}\label{sec:rewbl}
We review the syntax and semantics of \rewbl{}~\cite{Chida:2022,Chida:2023} step-by-step below.

\subsection{Regular Expressions with Backreferences and Lookaheads}\label{sec:rewbl-defs}
We first give the syntax of \rewbl{} over an alphabet $\Sigma$ and variables $\mathcal{V}$:
\[
\begin{array}{lcl}
E  & ::= & \sigma\ \mid\ \epsilon\ \mid\ E + E\ \mid\ E\,E\ \mid\ E^* \\
& \mid & \underbrace{(E)_v}_{\text{capture}} \mid \underbrace{\var{v}}_{\text{backreference}} \mid \underbrace{?(E)}_{\text{positive lookahead}} \mid \underbrace{!(E) }_{\text{negative lookahead}}
\end{array}
\]
where $\sigma \in \Sigma$ and $v \in \mathcal{V}$ is a variable.
The first line defines classical regular expressions, REGEX.
We consider the following subclasses in this paper:
\emph{REWB} (REGEX with captures and backreferences),
\emph{\prewb} (REWB with positive lookaheads $?(E)$),
\emph{\nrewb} (REWB with negative lookaheads $!(E)$).

\paragraph*{Semantics of REGEX}
We first give a semantics for REGEX.
To accommodate variables and lookaheads, configurations for REGEX are 4-tuples $\tuple{R, w, p, \Lambda}$ where
\begin{itemize}
\item $R$ is an expression to be executed;
\item $w$ is an input string and will not change throughout computation;
\item $p$ is a 0-origin position on $w$ ($0 \leq p \leq |w|$). We write $w[p]$ for the symbol on the position $p$.
It should be noted that $p = |w|$ is allowed to represent that we consume all the input.
\item $\Lambda : \mathcal{V} \to \Sigma^*$ is an assignment from variables $\mathcal{V}$ to substrings of $w$.
\end{itemize}

We write $\mathcal{C}$ for the set of configurations.
To denote all the results obtained by computing,
we use a semantic function
$
\sem{\cdot} :: \mathcal{C} \to \mathcal{P}(\mathbb{N} \times (\mathcal{V} \to \Sigma^*))
$
where $\mathcal{P}(X)$ is the power set of $X$.
On $\sem{\tuple{R, w, p, \Lambda}} = \set{ \tuple{p_1, \Lambda_1}, \ldots, \tuple{p_n, \Lambda_n} }$,
each pair $\tuple{p_i, \Lambda_i}$ means that, after executing $R$ on $w$ from $p$ under $\Lambda$,
we move to the position $p_i$ and obtain an assignment $\Lambda_i$.
On the basis of this idea, we define a semantics for each rule of the REGEX part:
\[
\sem{\tuple{\sigma, w, p, \Lambda}} = \begin{cases}
\set{ \tuple{p+1, \Lambda} } & \text{if $p < |w|$ and $w[p] = \sigma$}, \\
\emptyset & \text{otherwise,}
\end{cases}
\qquad
\sem{\tuple{\epsilon, w, p, \Lambda}} = \set{ \tuple{p, \Lambda} },
\]
\[
\begin{array}{rcl}
\sem{\tuple{E_1 + E_2, w, p, \Lambda}} & = & \sem{\tuple{E_1, w, p, \Lambda}} \cup \sem{\tuple{E_2, w, p, \Lambda}}, \\[5pt]
\sem{\tuple{E_1 E_2, w, p, \Lambda}} & = & \displaystyle\bigcup_{ \tuple{p',\,\Lambda'} \in \sem{\tuple{E_1, w, p, \Lambda}} } \sem{\tuple{E_2, w, p', \Lambda'}}, \\[5pt]
\sem{\tuple{E^*, w, p, \Lambda}} & = & \displaystyle\bigcup\limits^\infty_{i = 0} \sem{\tuple{E^i, w, p, \Lambda}}
\quad \text{\small where $E^0 = \epsilon$, $E^i = \overbrace{E E \cdots E}^{\smash{i}}$.}
\end{array}
\]
We note that our semantic function $\sem{\tuple{E, w, p, \Lambda}}$ is inductively defined on the lexicographic ordering over the star height of $E$ and the expression size of $E$.
The start height and expression size of \rewbl{} is defined in the usual way.

We also note that each $\sem{\tuple{E, w, p, \lambda}}$ forms a finite set  because the value of each variable $x$ must be a substring of $w$,
and also $p$ is bounded as $0 \leq p \leq |w|$.

\paragraph*{Semantics of REWB}
A capture expression $(E)_x$ attempts to match the input string with $E$.
If it succeeds, the matched substring is stored in the variable $x$.
\[
\sem{\tuple{\,(E)_x, w, p, \Lambda\,}} = \bigl\{ \left\langle\ p', \Lambda'\bigl[\,x \mapsto w[p..p')\,\bigr]\ \right\rangle : \tuple{p', \Lambda'} \in \sem{\tuple{E, w, p, \Lambda}} \bigr\}
\]
where $w[p..q)$ is the string $w[p]\,w[p+1] \cdots w[q-1]$.

A backreference $\var{x}$ refers to the substring stored previously by evaluating some $(E)_x$.
\[
\sem{\tuple{\var{x}, w, p, \Lambda}} = \sem{\tuple{\Lambda(x), w, p, \Lambda}}.
\]

\paragraph*{Semantics of Lookaheads}

Positive lookaheads $?(E)$ run $E$ from the current input without consuming any input. 
Although the change in head position is undone after running $E$, the modification to variables in $E$ is not.
So, we can also call it destructive lookahead.
\[
\sem{\tuple{?(E), w, p, \Lambda}} = \bigl\{ \tuple{p, \Lambda'} : \tuple{p', \Lambda'} \in \sem{\tuple{E, w, p, \Lambda}} \bigr\}.
\]

Negative lookaheads $!(E)$ also run $E$ without consuming any input. If $E$ does not match anything, we invoke a continuation.
Compared with positive lookaheads $?(E)$, both the previous head position and the previous values of variables are recovered.
\[
\sem{\tuple{!(E), w, p, \Lambda}}  =
\begin{cases}
\set{ \tuple{p, \Lambda} } & \text{if}\ \sem{\tuple{E, w, p, \Lambda}} = \emptyset, \\
\emptyset &  \text{otherwise.}
\end{cases}
\]

Readers may wonder why positive lookaheads modify variables while negative ones do not.
The reason for this asymmetry is that: in negative lookaheads, all computations fail uniformly and so there is no suitable configuration for altering variables.
On the other hand, using the non-destructive property of the negative lookaheads, we can also define non-destructive positive lookaheads as $!(!(E))$.
This expression executes $E$ from the current position without any variable modifications.

\subparagraph*{Remark: Special character $\$$}
Using a negative lookahead, we define $\$ = !(\Sigma)$ to check the end of the input.
$\$$ is one of the most important applications of negative lookaheads.

\begin{definition}
The language of a \rewbl{} $E$, $L(E)$, is defined as follows:
\[
L(E) = \bigl\{ w : \tuple{p, \Lambda} \in \sem{\tuple{E, w, 0, \iota}},\,p = |w| \bigr\}, 
\qquad \text{where\ $\forall x \in \mathcal{V}.\,\iota(x) = \epsilon$}.
\]
\end{definition}
We can also consider another definition using $\gamma(x) = \bot$ instead of $\iota$,
where $\gamma$ indicates that all variables are initially undefined.
Although some real-world regular expression engines adopt that definition,
we adopt the above $\iota$-definition since it is tedious to initialize all variables $x$ using $(\epsilon)_x$.
The results discussed in this paper will not change regardless of which one is used.
There is another formalization that excludes labels that appear multiple times in a single group, for instance $(\var{x} \var{x})_x$.
However, since this paper is interested in language classes in the most expressive setting possible, no restrictions are placed on such labels.

\newcommand{\myedge}[1]{\ulcorner #1 \urcorner}

\paragraph*{\nlog-complete Language Accepted by \rewb}
Hartmanis and Mahaney proposed a decision problem called \emph{TAGAP}, which is the topological sorted version of the reachability problem of directed \emph{acyclic} graphs (DAG)~\cite{Hartmanis:1981}.
Let us consider a word $w = \myedge{x_1 \to y_1} \myedge{x_2 \to y_2}  \cdots \myedge{x_n \to y_n}$, which represents a DAG.
We call $w$ \emph{topologically sorted} if: for all pairs of $a \to b$ and $b \to c$ of $w$, $a \to b$ appears before $b \to c$ in $w$.
The following example represents a DAG and one of its topologically sorted representation:
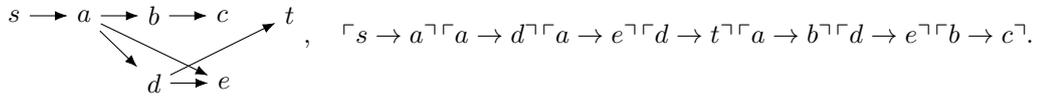
\[
\tikz[baseline=-10pt] {
\node (x) {$s$};
\node (a) [right = .5cm of x] {$a$};
\draw[-Latex] (x) edge (a);
\node (b) [right = .5cm of a] {$b$};
\draw[-Latex] (a) edge (b);
\node (c) [right = .5cm of b] {$c$};
\draw[-Latex] (b) edge (c);
\node (t) [right = .5cm of c] {$t$};
\node (d) [below = .4cm of b] {$d$};
\draw[-Latex] (a) edge (d);
\node (e) [right = .5cm of d] {$e$};
\draw[-Latex] (d) edge (e);
\draw[-Latex] (a) edge (e);
\draw[-Latex] (d) edge (t);
},\quad
\myedge{s \to a}\,\myedge{a \to d}\,\myedge{a \to e}\,\myedge{d \to t}\,\myedge{a \to b}\,\myedge{d\to e}\,\myedge{b \to c}.
\]

\vspace{-10pt}We define the language for TAGAP as follows:
\[
L_{\text{TAGAP}} = \set{ s\,\#\,\textit{R}\,\#\,t : \text{$R$ is a topologically sorted repr.~of $G$},\ \text{$t$ is reachable from $s$ in $G$} }.
\]

Hartmanis and Mahaney showed that this language is \nlog-complete~\cite[Theorem~3]{Hartmanis:1981}.
Since we only consider the topologically sorted representation, there is no longer a need to explore the entire input many times.
In the above example, we can reach $t$ from $s$ by nondeterministically finding edges $s \to a$, $a \to d$, and $d \to t$ in this order by one-way scanning.
Indeed, we can show the following theorem.
\begin{theorem}
The language class \rewblang{} contains the \nlog-complete language $L_{\text{TAGAP}}$.
\end{theorem}
\begin{proof}
The following expression $E_{\text{TAGAP}}$ clearly recognizes $L_{\text{TAGAP}}$:
\[
E_{\text{TAGAP}} = (V^*)_{\textsc{Cur}}\ \#\ \bigl(\Sigma^*\ \ulcorner\ \var{\textsc{Cur}}\,\to\,(V^*)_{\textsc{Cur}}\ \urcorner\ \bigr)^*\ \Sigma^*\,\#\ \var{\textsc{Cur}},
\]
where $V$ is an alphabet for vertices.
\end{proof}

\section{Expressiveness of \rewbl}\label{sec:expressiveness}

In this section, we present some theorems about the expressiveness of \rewbl.

\subsection{Unary Non-Indexed Language}\label{subsec:unary-non-indexed-languages}
We consider the single exponential numerical language $L_{\mathit{1exp}} = \set{ a^{2^k} : k \in \mathbb{N} }$ over the unary alphabet $\Sigma = \set{ a }$.
The language $L_{\mathit{1exp}}$ is represented by the \prewb{} expression $E_{\mathit{1exp}} = ?(a)_x\,(?(\var{x} \var{x})_x)^*\,\var{x}$.
The part $?(a)_x$ initializes $x = a$ (i.e., $x = a^0$), and the Kleene-$*$ part iteratively doubles $x$.

Furthermore, we can represent the doubly exponential language $L_{\mathit{2exp}} = \set{ a^{2^{2^k}} : k \in \mathbb{N} }$ by the following \rewbl{}-expression:
\[
E_{\mathit{2exp}} =
?((a)_m)\ \bigl(?[(\var{n}a)_n] ?[(\var{m}\var{m})_m]\bigr)^*\ 
?((a)_x)\ \bigl(?[(\var{y}a)_y] ?[(\var{x}\var{x})_x]\bigr)^*\ 
?(a^*\,?(\var{m}\$)\,?(\var{y}\$))\ 
\var{x}.
\]
It searches the numbers $n, m, x, y$ that satisfy $2^n = m$, $2^y = x$, and $m = y$; then, $x = 2^y = 2^{2^n}$.
While unfolding the first two Kleene-*, $m = 2^n$ and $y = 2^x$ hold.
The part $?(a^*\,?(\var{m}\$)\,?(\var{y}\$))$ checks if $m =_? y$ by effectively using the negative lookahead expression $\$ = !(a)$.

We emphasize the known result $L_{\mathit{2exp}} \notin \indexed$, which is shown by the pumping or shrinking lemma for indexed languages~\cite{Hayashi:1973,Gilman:1996}.
Since we can carry out a similar construction of $E_{\mathit{2exp}}$ in \prewb{} and \nrewb, we have the following result.
\begin{theorem}\label{thm:unary-non-indexed}
\prewb{} and \nrewb{} can represent unary non-indexed languages.
\end{theorem}
\begin{proof}[Proof (Sketch)]
Due to the page limitation, we provide a proof sketch for the \nrewb{} part and put the complete proof in the Appendix.
Let us consider the following \nrewb{} expression:
\[
E'_{\mathit{2exp}} = (a)_m\ E_1\ (a)_x\ E_2\ !(!(\var{m} \$))\ !(!(\var{y} \$))\ a^*
\]
where
$E_1 = ((\var{n} a)_n (\var{m} \var{m})_m)^*$ and
$E_2 = ((\var{y} a)_y (\var{x} \var{x})_x)^*$.
While the expression $E'_{\mathit{2exp}}$ resembles $E_{\mathit{2exp}}$, it lacks positive lookaheads.
Let us explain $E'_{\mathit{2exp}}$  step-by-step:
\begin{enumerate}
\item The expressions $(a)_m$ and $(a)_x$ initialize $m$ and $x$ by $a$ as with $E_{\mathit{2exp}}$.
\item The subexpression $E_1$ repeatedly expands variables $n$ and $m$ as with $E_{\mathit{2exp}}$.
So, executing $E_1$ \emph{actually} consumes inputs as follows without positive lookaheads:
\[
\underline{a^1}_n\ \underline{a^{2^1}}_m\ \underline{a^2}_n\ \underline{a^{2^2}}_m\ \underline{a^3}_n\ \underline{a^{2^3}}_m \cdots \underline{a^i}_n\ \underline{a^{2^i}}_m \cdots
\]
\item The same holds for the expression $E_2$.
\end{enumerate}

The part $!(!(\var{m}\$))$ is a non-destructive positive lookahead by $\var{m}\$$ that is simulated by double negative lookaheads;
therefore, the part $!(!(\var{m}\$))\ !(!(\var{y}\$))$ requires $m = y$. And finally, if we pass the assertion, we consume the rest input by $a^*$.
By replacing positive lookaheads with actual consuming, the language $L(E'_{\mathit{2exp}})$ \emph{grows faster} (the notion \emph{fast growth} will be formally introduced in the Appendix) than $L_{\mathit{2exp}}$. This property implies that $L(E'_{\mathit{2exp}})$ is not an indexed language.
\end{proof}

It states that, even if restricted to unary languages, positive or negative lookaheads make \rewb{} expressive and incomparable to \indexed{}.
We can also show the undecidability of the emptiness problem, which is checking if $L(E) = \emptyset$ for a given expression $E$, of \prewb{} and \nrewb{} over $\Sigma = \set{ a }$.
\begin{theorem}\label{thm:emptiness-of-unary}
The emptiness problems of \prewb{} over a unary alphabet and \nrewb{} over a unary alphabet are undecidable.
\end{theorem}
The two undecidability results can be shown by encoding the Post Correspondence Problem to the emptiness problems.
Due to the page limitation, we put the proof of Theorem~\ref{thm:unary-non-indexed} in the Appendix.

\subsection{Simulating Two-way Multihead Automata by \rewbl}

We show that \rewbl{} can simulate  \emph{two-way multihead automata}, which are a classical extension of automata and capture \nlog~\cite{Hartmanis:1972}.

\paragraph*{Simulating Two-way One-head Automata}
We start from two-way \emph{one}-head automata.
Let $\mathcal{A} = (Q, q_\text{init}, q_\text{acc}, \Sigma, \Delta)$ be a two-way automata where
$Q = \set{ q_0, q_1, \ldots, q_{|Q|-1} }$ is a finite set of states, $q_\text{init} \in Q$ is the initial state, $q_\text{acc} \in Q$ is the accepting state, $\Sigma$ is an input alphabet, and $\Delta$ is a set of transition rules.
Each transition rule is of the form $p \xrightarrow{\tau / \theta} q$ where $p, q \in Q$, $\tau \in (\Sigma \cup \set{\vdash, \dashv})$, and  $\theta \in \set{-1, 0, 1}$.
The component $\theta$ indicates a head moving direction:
(1) if $\theta = -1$ (resp. $\theta = 1$), we move the scanning head left (resp. right);
(2) if $\theta = 0$, we \emph{do not} move the scanning head.

For an input $w \in \Sigma^*$, we run $\mathcal{A}$ on the extended string $\vdash w \dashv$, which are surrounded by the left and right end markers.
A configuration of $\mathcal{A}$ for $\vdash w \dashv$ is a tuple $(q, i)$ where $q \in Q$ and $0 \leq i < 2 + |w|$,
and thus the current scanning symbol is $(\vdash w \dashv)[i]$.

A transition rule $\delta = p \xrightarrow{\tau / \theta} q \in \Delta$ gives a labelled transition relation $\overset{\smash{\delta}}{\Rightarrow}$ as follows:
\[
(p, i) \overset{\smash{\delta}}{\Rightarrow} (q, i + \theta) \quad \text{if $(\vdash w \dashv)[i] = \tau$ and $0 \leq i + \theta < |w|+2$}
\]

The word $w \in \Sigma^*$ is accepted by $\mathcal{A}$ if the initial configuration $(q_\text{init}, 0)$, reading $\vdash$, has a computation path to a configuration with the accepting state $q_\text{acc}$. We now define the language of $\mathcal{A}$ as follows:
\[
L(\mathcal{A}) = \set{ w : (q_\text{init}, 0) \overset{\smash{\delta_1}}{\Rightarrow} (q_1, i_1)  \overset{\smash{\delta_2}}{\Rightarrow} \cdots  \overset{\smash{\delta_n}}{\Rightarrow} (q_\text{acc}, i_n) }.
\]

To simulate $\mathcal{A}$ by \rewbl{}, we use some variables $L, R, S \in \Sigma^*$, and $C \in \Sigma \cup \set{ \epsilon }$.
Intuitively, each variable means the following:
\begin{itemize}
\item If $C \in \Sigma$, it means that $\mathcal{A}$ is scanning $C$. If $C = \epsilon$, it means that $\mathcal{A}$ is located on $\vdash$ or $\dashv$.
\item $L$ (resp. $R$) denotes the substring left (resp. right) to $C$; so, $w = L C R$ is an invariant.
\item The length of $S$ denotes the index $i$ of the current state $q_i$ of $\mathcal{A}$.
\end{itemize}
We formalize the above intuition as the following simulation $\sim$ between $(q, i)$ and $\tuple{L, C, R, S}$:
\[
\begin{array}{ll}
(q_j, 0) \sim \tuple{\epsilon, \epsilon, w, S} & \text{if $|S| = j$}, \\
(q_j, |w|+1) \sim \tuple{w, \epsilon, \epsilon, S} & \text{if $|S| = j$}, \\
(q_j, i) \sim \tuple{L, C, R, S} & \text{if $w = L\,C\,R$, $|L| = i-1$, $C = (\vdash w \dashv)[i]$, and $|S| = j$}.
\end{array}
\]
To represent all states, we need $| w | \geq |Q|-1$.
So, we mainly consider to represent $L(\mathcal{A}) \setminus L'$ where $L' = \set{ w \in L(\mathcal{A}) : | w | < |Q|-1 }$.
For instance, our simulation proceeds as follows:
\[
\begin{array}{ll@{\ \,}c@{\ \,}l@{\ }l}
(1) & \vdash^{q_0} a b \dashv & \sim & \tuple{L = \epsilon,\ C = \epsilon, R = a b,\ S = \epsilon}, & \text{\small (if $R = w$, then $L = C = \epsilon$ and $\mathcal{A}$ is on $\vdash$)}, \\
(2) & \vdash a^{q_1} b \dashv & \sim & \tuple{L = \epsilon,\ C = a,\ R = b,\ S = a}, & \text{\small (move right from (1))} \\
(3) & \vdash a b^{q_1} \dashv & \sim & \tuple{L = a,\ C = b,\ R = \epsilon,\ S = a}, \\
(4) & \vdash a b \dashv^{q_1} & \sim & \tuple{L = a b,\ C = \epsilon,\ R = \epsilon,\ S = a} & \text{\small (if $L = w$, then $R = C = \epsilon$ and $\mathcal{A}$ is on $\dashv$)}, \\
(5) & \vdash a b^{q_2} \dashv & \sim & \tuple{L = a,\ C = b,\ R = \epsilon,\ S = ab}, & \text{\small (move left from (5))} \\
(6) & \vdash a^{q_2} b \dashv & \sim & \tuple{L = \epsilon,\ C = a,\ R = b,\ S = ab}.
\end{array}
\]

As initialization, we use the expression $E_\text{init} =\ (\epsilon)_{L} (\epsilon)_{C}\ ?((\Sigma^*)_{R}\,\$)\ (\epsilon)_S$,
which sets $L = C = S = \epsilon$ and $R = w$ for the input $w$.

To move the head right, we use the following expression
\[
E_{+1} =\ !(\var{L}\$)\ ?\bigl[(\var{L} \var{C})_L (\Sigma)_C (\Sigma^*)_R \$\ +\ (\var{L} \var{C})_L (\epsilon)_C (\epsilon)_R \$\bigr]
\]
where the part $!(\var{L}\$)$ checks if the right end marker $\dashv$ is being scanned.
Similarly, we define the expression $E_{-1}$ to move the head left as follows:
\[
E_{-1} =\ !(\var{R}\$)\ ?\big[ (\Sigma^*)_L (\Sigma)_C (\var{C} \var{R})_R \$ + (\epsilon)_L (\epsilon)_C (\var{C} \var{R})_R \$ \bigr].
\]

To check the current scanning symbol is $\sigma \in \Sigma$, $\vdash$, or $\dashv$, we use the following expression:
\[
E_{\sigma} = !(\var{L} \$) !(\var{R} \$)\,?(\var{L}\, \sigma\, \var{R} \$),\ 
E_{\dashv} = ?(\var{L} \$),\  
E_{\vdash} = ?(\var{R} \$).
\]

To check if the current state is $q_i$, we use the expression $E_{q_i} = ?\bigl(\Sigma^*\ ?(\var{S}) \$\ ?(\Sigma^i) \$\bigr)$.
To change the current state $q_i$ to $q_j$, we use the expression
$E_\text{$i$-to-$j$} = E_{q_i}\, ?(\Sigma^*\,(\Sigma^j)_{S}\,\$)$.

Now, each transition rule $\delta$ is simulated by the following expression $E(\delta)$ defined as:
\[
E(q_i \xrightarrow{\smash{\tau / 0}} q_j) = E_\text{$i$-to-$j$}\ E_{\tau},\quad
E(q_i \xrightarrow{\smash{\tau / \theta}} q_j) = E_\text{$i$-to-$j$}\ E_{\tau}\ E_\theta~~(\theta \in \set{-1, +1}).
\]

Finally, the following expression $E_\mathcal{A}$ simulates $\mathcal{A}$, and $L(E_\mathcal{A}) = L(\mathcal{A})$ holds:
\[
E_\mathcal{A} = E_{L'}\ +\ ?(E_{q_\text{init}}\,(E(\delta_1) + E(\delta_2) + \cdots + E(\delta_n))^*\,E_{q_\text{acc}})\,\Sigma^*
\]
where $E_{L'}$ is a regular expression for the finite language $L'$, and $\Delta = \set{\delta_1, \delta_2, \ldots, \delta_n}$.

\paragraph*{Simulating Two-way Multihead Automata}

We extend the above argument to two-way \emph{multihead} automata $\mathcal{M}$~\cite{Hartmanis:1972,Holzer:2011}.
Compared to two-way one-head automata $\mathcal{A}$, $\mathcal{M}$ has multiple-heads on input strings.
We write $K$ for the number of heads. The difference between $\mathcal{A}$ and $\mathcal{M}$ are the following:
\begin{itemize}
\item Each configuration of $\mathcal{M}$ is a tuple $(q, i_1, i_2, \ldots, i_K)$ where $q$ is the current state and $i_j$ is the $j$-th head position.
\item Each transition rule is $p \xrightarrow{ (\tau_1, \ldots, \tau_K) / (\theta_1, \ldots, \theta_K) } q$ where $p, q \in Q$, 
$\tau_j \in \Sigma \cup \set{ \vdash, \dashv }$ is used for inspecting the scanned symbol by $j$-th head,
and $\theta_j$ denotes the head moving direction for the $j$-th head.
\end{itemize}
We define a transition relation $\Rightarrow$ in the same way as for $\mathcal{A}$.
Let $\delta = p \xrightarrow{ (\tau_1, \ldots, \tau_K) / (\theta_1, \ldots, \theta_K) } q$ be a rule
and $C = (p, i_1, i_2, \ldots, i_K)$ be a valid configuration.
If $\forall 1 \leq j \leq K.\,(\vdash w \dashv)[i_j] = \tau_j$, then we have $C \overset{\delta}{\Rightarrow} (q, i_1 + \theta_1, i_2 + \theta_2, \ldots, i_K + \theta_K)$.
We also define the language in the same way:
\[
L(\mathcal{M}) = \set{ w  : (q_\text{init}, 0, 0, \ldots, 0) \overset{\delta_1}{\Rightarrow}  \overset{\delta_2}{\Rightarrow} \cdots  \overset{\delta_n}{\Rightarrow} (q_\text{acc}, \ldots) }.
\]

Since we can easily extend our above construction for two-way multihead automata,
we give an expression $E_{\mathcal{M}}$ for a given multihead automata $\mathcal{M}$ such that $L(E_\mathcal{M}) = L(\mathcal{M})$.

It is well-known that the class of languages accepted by two-way multihead automata corresponds to \nlog~\cite{Hartmanis:1972}; so, we have the following theorem.
\begin{theorem}\label{thm:simulate-2mfa}
$\nlog \subseteq \rewbllang$.
\end{theorem}

\subsection{True Quantified Boolean Formula}

We translate the \pspace-complete problem \textbf{TQBF}, checking if a \emph{quantified boolean formula} (QBF) is true, into the membership problem of \rewbl.
Here we only consider QBFs in CNF since TQBF restricted to CNF is \pspace-complete~\cite{Arora:2009}.
For instance, let us consider the following QBF $Q$ and translate it to the equivalent form by replacing $\forall$ with $\neg \exists \neg$:
\[
Q = \forall a.\,\exists b.\,\forall c. \forall d. \,(a \lor b \lor c) \land (\overline{b} \lor c \lor d) \Longrightarrow
\neg \exists a.(\neg \exists b.(\neg \exists c.(\exists d.(\neg ((a \lor b \lor c) \land (\overline{b} \lor c \lor d)) )))).
\]
In order to check if $Q$ is true, we first structurally translate $Q$ into the following $E_Q$:
\[
\begin{array}{l}
E(v) = \bigl((T)_v (F)_{\overline{v}}\ + (T)_{\overline{v}} (F)_{v}\bigr) \quad \text{where $v$ is a propositional variable}, \\
E_Q\,=\,!(E(a)\ !(E(b)\ !(E(c)\ E(d)\ !( (\var{a} + \var{b} + \var{c})(\var{\overline{b}} + \var{c} + \var{d}) )))),
\end{array}
\]
and then check $(w =) T F\,T F\,T F\, T F\, T T \in_? L(E_Q)$.
We explain the string $w$ using the annotated version $T_1 F_2 T_3 F_4 T_5 F_6 T_7 F_8 T_9 T_{10}$:
(1) the first two characters $T_1 F_2$ makes the two cases where $(a = T, \overline{a} = F)$ or $(a = F, \overline{a} = T)$;
(2) similarly, $T_3 F_4$ (resp. $T_5 F_6$ and $T_7 F_8$) works for $b$ and $\overline{b}$ (resp. $c, \overline{c}$ and $d, \overline{d}$) ;
(3) by $T_9$, we check if the expression $(a \lor b \lor c)$ holds (in the negative context);
(4) by $T_{10}$, we also check if $(\overline{b} \lor c \lor d)$ holds.
Thus, $Q$ is true iff $w \in L(E_Q)$.
We remark that $Q$ is false because, for $a = F$, $b$ should be $T$ for the first clause; however, the second clause fails for $c = T$ and $d = T$.

On the basis of the above translation using $E(v)$,
we can translate every CNF-QBF $Q$ to the corresponding expression $E_Q$
and give the membership problem $TFTF \ldots TF\ TT\ldots T \in_? L(E_Q)$ in polynomial time for the size of $Q$.
It implies the following result.
\begin{theorem}\label{thm:tqbf-and-rewbl}
The membership problem of \rewbl{} is \pspace-hard.
\end{theorem}

\section{Log-space Nested-Oracles Nondeterministic Turing Machines}\label{sec:nested-oracle-logspace-ntm}
As we have stated in the Introduction, we utilize log-space nested-oracles NTMs.
We will translate \rewbl{} to them in the next section.

We first review log-space NTMs.
Here we especially consider \emph{$c$-bounded $k$-working-tapes} log-space NTM $M = (k, c, Q, q_\text{init}, Q_F, \Sigma, \Gamma, \Box, \Delta)$.
Each component of $M$ means:
\begin{itemize}
\item $k$ is the number of working tapes $T_1, T_2, \ldots, T_k$.
\item $c$ is used to bound the size of working tapes. It will be defined precisely below.
\item $Q$ is a finite set of states, $q_\text{init}$ is the initial state, and $Q_F \subseteq Q$ is a set of accepting states.
\item $\Sigma$ is an input alphabet.
\item $\Gamma$ is a working tape alphabet. $\Box \in \Gamma$ is the blank symbol for working tapes.
\item $\Delta$ is a set of transition rules; Each rule is either 
$p \xrightarrow{\tau\,\mid\,\theta} q$\ \ or\ \ $p \xrightarrow[T_i]{\kappa\ \mapsto\ \kappa'\,\mid\,\theta} q$
\\
where $p, q \in Q$, $\tau \in \Sigma \cup \set{ \vdash, \dashv }$, $\kappa, \kappa' \in \Gamma \cup \set{ \vdash, \dashv }$, and $\theta \in \set{-1, +1, 0}$.
\end{itemize}

Let $w \in (\vdash \Sigma^* \dashv)$ be a string surrounded by the left and right end markers.
Valid configurations of $M$ for $\vdash w \dashv$ are tuples $\tuple{q, i, (T_1, i_1), \ldots, (T_k, i_k)}$ where
\begin{itemize}
\item $q \in Q$ is the current state. $i \in \mathbb{N}$ ($0 \leq i < |w|+2$) is the current head position on $\vdash w \dashv$.
\item $T_x \in (\vdash \Gamma^C \dashv)$ where $C = c \cdot \ilog{|w|}$ is the $x$-th working tape surrounded by the end markers.
$\lceil \cdot \rceil$ is the ceiling function to integers; for example, $\ilog{3} = \lceil 1.584\ldots \rceil = 2$.
\textbf{Remark:} The tape capacity $C$ is determined by the parameter $c$ and the input $w$.
\item $i_x$ is the $x$-th tape head position on $T_x$ ($0 \leq i_x < C + 2$).
\end{itemize}

We write $\textbf{Valid}_M(w)$ (or, simply $\textbf{Valid}(w)$) for the set of valid configurations for the input $w$.
It is clear that $|\textbf{Valid}(w)| = |Q| \times (|w| + 2) \times (|\Gamma|^C \times (C + 2))^k$ where $C = c \cdot \ilog{|w|}$.

For an input string $w$,
we write $\mathcal{I}(w)$ to denote the initial configuration on $\vdash w \dashv$:
\[
\mathcal{I}(w) = \tuple{q_{\text{init}},\ 0,\ (\vdash \Box^C \dashv, 0), \ldots, (\vdash \Box^C \dashv, 0)} \text{ where } C = c \cdot \ilog{|w|}.
\]

Let $\xi = \tuple{p, i, (T_1, i_1) \ldots, (T_x, i_x), \ldots, (T_k, i_k)}$ be a valid configuration on $\vdash w \dashv$.
For each transition rule $\delta$, we define a labelled transition relation $\overset{\smash{\delta}}{\Rightarrow}$ on \emph{valid} configurations as follows:
\[\hspace{-7pt}
\AXC{$p \xrightarrow[\phantom{I}]{\tau\,\mid\,\theta} q \in \Delta$}
\AXC{$(\vdash w \dashv)[i] = \tau$}
\BIC{$\xi \overset{\delta}{\Rightarrow} \tuple{q,\ i + \theta,\ (T_1, i_1), \ldots, (T_k, i_k)}$}
\DisplayProof\quad
\AXC{$p \xrightarrow[\smash{T_x}]{\kappa\,\mapsto\,\kappa'\ \mid\ \theta} q \in \Delta$}
\AXC{$\kappa = T_x[i_x]$}
\BIC{$\xi \overset{\delta}{\Rightarrow} \tuple{q, i, (T_1, i_1), \ldots, (T_x[i_x] := \kappa', i_x + \theta), \ldots, (T_k, i_k)}$}
\DisplayProof
\]
where $T_x[i_x] := \kappa'$ is the new working tape obtained by writing $\kappa'$ to the position $i_x$.

We write $\nlog(c, k)$ for the set of $c$-bounded $k$-working-tapes log-space NTMs.
If $c$ and $k$ is not important, by abusing notation, we simply write $\nlog$.
For $M \in \nlog$ and an input string $w$, we write $M(w, \xi)$ to denote the set of valid and acceptable configurations that are reachable from a valid configuration $\xi$ on $\vdash w \dashv$:
\[
M(w, \xi) = \set{ \xi'\ :\ \xi \Rightarrow^* \xi',\,\xi' = \tuple{q_\text{acc}, i, \mathcal{T}},\ \text{$q_\text{acc}$ is an accepting state of $M$} },
\]
where $\mathcal{T}$ is a sequence of pairs of a working tape and an index $(T_1, i_1) \ldots (T_k, i_k)$.
Now the language $L(M)$ is defined as: $L(M) = \set{ w : M(w, \mathcal{I}(w)) \neq \emptyset }$.

Here we state a useful proposition, which will be used below sometimes.
\begin{proposition}\label{prop:space-for-configurations}
Let $M \in \nlog(c, k)$. For any input $w$,
to represent each valid configuration or equivalently store $|\textbf{Valid}(w)|$,
we need an extra $O(c \cdot k)$-bounded working tape.
\end{proposition}
\begin{proof}
Please recall that $|\textbf{Valid}(w)| = |Q| \times (|w| + 2) \times (|\Gamma|^C \times (C + 2))^k$ where $C = c \cdot \ilog{|w|}$.
So, $\log |\textbf{Valid}(w)| = (k \cdot c \cdot \log |w|) \log|\Gamma| + \cdots = O(c \cdot k) \cdot \log |w|$.
Thus, we need an $O(c \cdot k)$-bounded tape.
\end{proof}

On log-space NTMs, we can solve the problem-\textit{\textbf{(B)}} in Section \ref{sec:intro:challenge}.
\begin{proposition}\label{prop:finite-computation-tree}
Let $M \in \nlog(c, k)$. There exists $N \in \nlog(O(c \cdot k), k+1)$ such that $L(M) = L(N)$ and,
for any input $w$, all computations of $N$ starting from $w$ eventually halt.
\end{proposition}
\begin{proof}
The number of reachable configurations is bounded by $\mathcal{B} = |\mathbf{Valid}_M(w)|$.
So, we can safely ignore all paths $P$ whose length $> \mathcal{B}$ without changing the accepting language.
To check if the current path length $> \mathcal{B}$, we need an $O(c \cdot k)$-bounded tape by Proposition~\ref{prop:space-for-configurations}.
\end{proof}

We note that properties, like Proposition~\ref{prop:finite-computation-tree}, are insufficient to show that a language class is closed under complement.\footnote{
For example, we can translate any nondeterministic pushdown automata to real-time ones, which do not have $\epsilon$-transitions.
However, the class of context free languages is not closed under complement.}
Thus, the problem-(\textit{\textbf{A}}) in Section \ref{sec:intro:challenge} is essentially hard; indeed, it is the interesting part of Immerman--Szelepcs{\'e}nyi theorem~\cite{Immerman:1988,Szelepcsenyi:1988}.
In the following subsection, we revisit their theorem along with introducing log-space nested-oracles NTM.

\subsection{Augmenting NTM with Nested Oracles}

We extend log-space NTM with finitely nested oracles (or subroutines) to naturally handle nested lookaheads of \rewbl.
Similar to our definition of log-space NTMs, we consider $c$-bounded $k$-tapes log-space nested oracles NTMs.

To allow nested oracle calling, we inductively define our machines.
First, as the base case, we write $\oraclem^0(c, k) = \nlog(c, k)$ to denote machines without oracles.
Next, as the induction step, we define $\oraclem^{x+1}(c, k)$ using $\oraclem^x(c, k)$ as follows.
Each $M \in \oraclem^{x+1}(c, k)$ is a tuple $(k, c, Q, q_{\text{init}}, Q_F, \Sigma, \Gamma, \Box, \Delta)$.
The only difference from log-space NTMs are two new types of transition rules, called \emph{oracle transition rules}:
\begin{itemize}
\item $p \xrightarrow{\in N} q$: asking an oracle $N \in \oraclem^y(c, k)$ where $y \leq x$ in the positive context.
\item  $p \xrightarrow{\not\in N} r$: asking an oracle $N \in \oraclem^y(c, k)$ where $y \leq x$ in the negative context.
\end{itemize}

We write $\oraclem^{\omega}(c, k)$ for $\bigcup^{\infty}_{i = 0} \oraclem^i(c, k)$.
To denote the nesting level of machines $M \in \oraclem^{\omega}(c, k)$,
we inductively define the function $\depth$ as follows:
\[
\begin{array}{l}
\depth(M) = \begin{cases}
0 & \text{if $M \in \oraclem^0(c, k)$}, \\
1 + \max\set{ \depth(N) : p \xrightarrow{\in N} q,\ p'\xrightarrow{\notin N} q \in \Delta(M) } & \text{otherwise}.
\end{cases}
\end{array}
\]
If $c$ and $k$ are not important, we simply write as $\oraclem^n$ and $\oraclem^\omega$.

Now, we define the semantics of oracle transitions as follows:
\[
\AXC{$p \xrightarrow{ \in N } q$}
\AXC{$N(w, \tuple{q^N_{\text{init}}, i, \mathcal{T}}) \ni \tuple{r,\, j,\ \mathcal{U}}$}
\BIC{$\tuple{p, i, \mathcal{T}} \Rightarrow \tuple{q, i, \mathcal{U}}$}
\DisplayProof
\qquad
\AXC{$p \xrightarrow{ \notin N } r$}
\AXC{$N(w, \tuple{q^N_{\text{init}}, i, \mathcal{T}}) = \emptyset$}
\BIC{$\tuple{p, i, \mathcal{T}} \Rightarrow \tuple{r, i, \mathcal{T}}$}
\DisplayProof
\]
where $q^N_\text{init}$ is the initial state of $N$, $\mathcal{T}$ and $\mathcal{U}$ is a sequence of pairs of a working tape and an index $(T_1, i_1) \ldots (T_k, i_k)$,
and the function $N(w, \xi) = \set{ \xi' : \xi \Rightarrow^* \xi', \xi' = \tuple{q, i, \mathcal{T}}, q \in Q_F(M) }$ is defined inductively on the depth of machines.

The semantics of $p \xrightarrow{\in N} q$ means that:
(1) we call an oracle (subroutine) $N$ for $\vdash w \dashv$ with the current position $i$ and the current working tapes $\mathcal{T}$ as its initial working tapes;
and (2) if $N$ accepts $w$, then we enter a state $q$ with the original position $i$ and working tapes $\mathcal{U}$ of $N$'s accepting configuration.
The semantics of $p \xrightarrow{\notin N} q$  means that, if $N$ does not accept $w$, we enter a state $r$ with the original position and working tapes $\mathcal{T}$.\footnote{
For simplicity, our oracle formalization differs from traditional treatments [16, 25, 3, 21] in some points: (1) we omit the use of oracle tapes, and (2) we allow inheriting configurations from called oracles. Despite these differences, our definition is adequate for Theorem 13 and for \rewbl{} in Section 6.}

\begin{example}\label{log-space-example1}
Using log-space nested-oracles NTMs, we simulate the following \rewbl:
\[
E_{\text{reach}} = (V^*)_{\textsc{Cur}}\,\#\ \bigl(?(\Sigma^*\,\#\,\var{\textsc{Cur}}\to(V^*)_{\textsc{Cur}}\,\#)\bigr)^*\,\Sigma^*\,\#\,\var{\textsc{Cur}}.
\]

For the sake of simplicity, we assume that $V = \set{ \star }$ is a unary alphabet.
The subexpression $(\Sigma^*\,\#\,\var{\textsc{Cur}}\to(V^*)_{\textsc{Cur}}\,\#)$ is simulated by the following log-space NTM $M \in \nlog$:
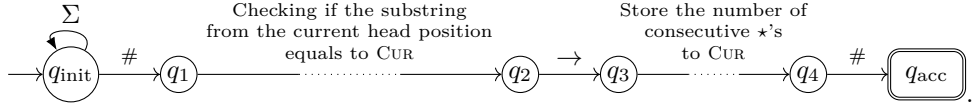
\[
\tikz {
\tikzset{largen/.style={fill=white, draw=black, rectangle, rounded corners=4pt, inner sep=6pt}}

\node[state, initial, initial text={}, fill=white, draw=black, minimum size=0pt, inner sep=1pt]
(qinit) {$q_\text{init}$};

\draw[-Latex] (qinit) edge [loop above, out=60, in=120, looseness=2.5] node {$\Sigma$} (qinit);

\node[state, fill=white, draw=black, minimum size=2pt, inner sep=1pt]
(q1) [right = .8cm of qinit] {$q_1$};

\draw[->] (qinit) edge[auto] node{\scriptsize $\#$} (q1);

\node[state, fill=white, draw=black, minimum size=2pt, inner sep=1pt]
(q2) [right = 4cm of q1] {$q_2$};

\draw[middle dotted line] (q1) -- node[above]{\scriptsize \begin{tabular}{c}Checking if the substring \\ from the current head position \\ equals to \textsc{Cur}\end{tabular}} (q2);

\node[state, fill=white, draw=black, minimum size=2pt, inner sep=1pt]
(q3) [right = .8cm of q2] {$q_3$};

\draw[->] (q2) edge[auto] node{$\to$} (q3);

\node[state, fill=white, draw=black, minimum size=2pt, inner sep=1pt]
(q4) [right = 2cm of q3] {$q_4$};

\draw[middle dotted line] (q3) -- node[above]{\scriptsize \begin{tabular}{c}Store the number of \\ consecutive $\star$'s \\ to \textsc{Cur}\end{tabular}} (q4);

\node[accepting, largen]
(q5) [right = .8cm of q4] {$q_\text{acc}$};

\draw[->] (q4) edge[auto] node{\scriptsize $\#$} (q5);
}.
\]
Using $M$, we give the following machine $N \in \oraclem^1$, clearly accepting $L(E_\text{reach})$:
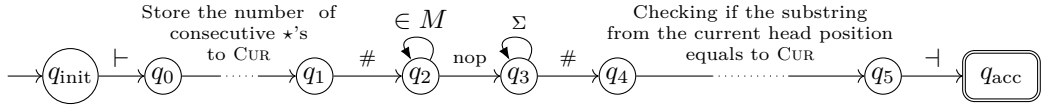
\[
\tikz {
\tikzset{largen/.style={fill=white, draw=black, rectangle, rounded corners=4pt, inner sep=6pt}}

\node[state, initial, initial text={}, fill=white, draw=black, minimum size=0pt, inner sep=1pt] (qinitpre) {$q_\text{init}$};

\node[state, fill=white, draw=black, minimum size=0pt, inner sep=1pt]
(qinit) [right = .6cm of qinitpre] {$q_0$};

\node[state, fill=white, draw=black, minimum size=2pt, inner sep=1pt]
(q1) [right = 1.5 cm of qinit] {$q_1$};

\draw[->] (qinitpre) edge[auto] node{$\vdash$} (qinit);

\draw[middle dotted line] (qinit) -- node[above]{\scriptsize  \begin{tabular}{c}Store the number ~of \\ consecutive $\star$'s \\ to \textsc{Cur}\end{tabular}} (q1);

\node[state, fill=white, draw=black, minimum size=2pt, inner sep=1pt]
(q2) [right = .9cm of q1] {$q_2$};

\draw[->] (q1) edge[auto] node{\scriptsize $\#$} (q2);

\draw[-Latex] (q2) edge [loop above, out=60, in=120, looseness=4.5] node {$\in M$} (q2);

\node[state, fill=white, draw=black, minimum size=2pt, inner sep=1pt]
(q3) [right = .8cm of q2] {$q_3$};

\draw[->] (q2) edge [above] node {\scriptsize nop} (q3);

\node[state, fill=white, draw=black, minimum size=2pt, inner sep=1pt]
(q4) [right = .8cm of q3] {$q_4$};

\draw[->] (q3) edge[auto] node{\scriptsize $\#$} (q4);
\draw[-Latex] (q3) edge [loop above, out=60, in=120, looseness=4.5] node {\scriptsize $\Sigma$} (q3);

\node[state, fill=white, draw=black, minimum size=2pt, inner sep=1pt]
(q5) [right = 3cm of q4] {$q_5$};

\draw[middle dotted line] (q4) -- node[above]{\scriptsize \begin{tabular}{c}Checking if the substring \\ from the current head position \\ equals to \textsc{Cur}\end{tabular}} (q5);

\node[accepting, largen] (q6) [right = .8 cm of q5] {$q_\text{acc}$};

\draw[->] (q5) edge[auto] node{$\dashv$} (q6);
}
\]
where edges labelled with ``nop'', $\overset{\text{nop}}{\longrightarrow}$, mean transitions that only change states and do not depend scanning symbol.
It just a syntax sugar because we can define  $p \overset{\text{nop}}{\longrightarrow} q$ by a set of transition rules $\set{ p \xrightarrow{ \tau\,\mid\,0 } q : \tau \in \Sigma \cup \set{ \vdash, \dashv } }$.
\end{example}
\begin{example}
We can also accept the language of non-reachability problems:
\[
L_{\text{non-reach}} = \set{ s\,\#\,x_1 \to y_1 \# \cdots \# x_n \to y_n \# t : \text{there is no path from $s$ to $t$} }.
\]
To recognize this language, we use the following $M_\text{non-reach} \in \oraclem^{2}$: \\
\tikz[baseline=-3pt] {
\tikzset{largen/.style={fill=white, draw=black, rectangle, rounded corners=4pt, inner sep=6pt}}

\node[state, initial, initial text={}, fill=white, draw=black, minimum size=0pt, inner sep=1pt] (qinitpre) {$q_\text{init}$};

\node[state, fill=white, draw=black, minimum size=0pt, inner sep=1pt]
(qinit) [right = .8cm of qinitpre] {$q_0$};

\draw[->] (qinitpre) edge[auto] node{$\vdash$} (qinit);

\node[state, fill=white, draw=black, minimum size=2pt, inner sep=1pt]
(q1) [right = 1 cm of qinit] {$q_1$};

\draw[->] (qinit) edge[auto] node{\small $\in N''$} (q1);

\node[accepting, largen] (q2) [right = 1 cm of q1] {$q_\text{acc}$};

\draw[->] (q1) edge[auto] node{\small $\notin N$} (q2);
}
where $N'' \in \nlog$ recognizes the language represented by the expression $V^*\,\#\,(V^* \to V^* \#)^*\,V^*$.
\end{example}

\subsection{Collapsing $\oraclem^\omega$ by Immerman--Szelepcs{\'e}nyi theorem}

Thanks to Proposition~\ref{prop:finite-computation-tree}, we can give a decision procedure to check $w \in_? L(M_{\text{non-reach}})$ for our above examples.
However, it is not clear that $\oraclem^2 =_? \nlog$ and more generally $\oraclem^\omega =_? \nlog$.
For example, is there a log-space NTM $N$ that recognizes $L_{\text{non-reach}}$\,?

Fortunately, the class \nlog{} is closed under complement, $\nlog = \textbf{co-}\nlog$.
This result is known as Immerman--Szelepcs{\'e}nyi theorem~\cite{Immerman:1988,Szelepcsenyi:1988}.
We employ their proof to collapse $\oraclem^x$ for some $x$ to log-space NTMs $\oraclem^0 = \nlog$.

\def\immerman#1{\overline{#1}}

\begin{lemma}\label{lemma:immermans-construction}
Let $M \in \nlog(c, k)$ for some $c$ and $k$.
There is a machine $\immerman{M} \in \nlog(O(c \cdot k), k + \partial)$
where $L(\immerman{M}) = \Sigma^*\!\setminus\!L(M)$ and $\partial$ is independent of $M$, $c$, and $k$.
\end{lemma}
\begin{proof}
We review Immerman's original construction in \cite{Immerman:1988}.
For the reader who would like to know more detailed explanation about his construction, we recommend some literature~\cite{sipser}.
His construction consists of the following two parts:
\begin{itemize}
\item Let \textsc{Start} be a configuration of $M$.
First, we compute the number $C$, the total number of configurations reachable from \textsc{Start}.
\item Next, using $C$,  we check if there is a path from \textsc{Start} to an acceptable configuration.
\end{itemize}

\def\conf{\mathbf{Valid}_M(w)}
The first part is accomplished by the following pseudocode~\cite[Lemma~2]{Immerman:1988}.
\begin{lstlisting}[basicstyle=\ttfamily,mathescape,columns=spaceflexible]
global $w$; // input string

def one_step$_M$($x$, $x'$): // $x$ and $x'$ $\text{\small are valid configurations}$
    foreach transition rule $\delta \in \Delta(M)$: // $\Delta(M)$ $\text{\small is the set of transition rules}$
        if $x \overset{\smash{\delta}}{\Rightarrow} x'$: return True;
    return False;

def counting$_M$($\textsc{Start}$):
    $\mathit{cur} \gets 1$; // $\text{\small the number of reachable configurations within}$ $\leq \mathit{dist}$ $\text{\small steps}$.
    for($\mathit{dist} \gets 0$,$\mathit{next} \gets 0$; $\mathit{dist} < |\conf|$; $\mathit{dist}$ += 1,$\mathit{cur} \gets \mathit{next}$,$\mathit{next} \gets 0$):
        foreach $x \in \conf$:
            $count \gets 0$; $\mathit{found\_x} \gets$ false;
            foreach $y \in \conf$: 
                $z \gets \textsc{Start}$; // $\text{\small search a path from}$ $\textsc{Start}$ $\text{\small to}$ $y$ (size $\leq \mathit{dist}$)
                for($i \gets 0$; $z \neq y$ & $i < \mathit{dist}$; $i$ += 1):
                    $z' \gets \text{Nondeterministically generated configuration}$;
                    if one-step($z$, $z'$): $z \gets z'$;
                    else: break;
                if $z = y$:
                    $\mathit{count}$ += $1$;
                    if one-step($y$, $x$): {  $\mathit{next}$ += 1; $\mathit{found\_x} \gets$ true; break; }
            if $\neg\mathit{found\_x}$ & $\mathit{count} \neq \mathit{cur}$: Halt and reject;
    return $\mathit{cur}$;
\end{lstlisting}
This function requires extra working tapes at least for $\mathit{cur}, \mathit{dist}, \mathit{next}, x, \mathit{count}, y, i, z, z'$.
By Proposition~\ref{prop:space-for-configurations}, for these variables, we need $O(c \cdot k)$-bound working tapes.

The second part is accomplished by the following pseudocode~\cite[Lemma~1]{Immerman:1988}.
\begin{lstlisting}[basicstyle=\ttfamily,mathescape,columns=spaceflexible]
def judge$_M$($\textsc{Start}$):
    $C \gets $ counting$_M$($\textsc{Start}$);
    $\mathit{count} \gets 0$;
    foreach $x \in \conf$:
        $y \gets \textsc{Start}$;
        for($i \gets 0$; $i \leq C$; $i$ += $1$):
            $y' \gets \text{Nondeterministically generated configuration}$;
            if one-step($y$, $y'$):
                $y \gets y'$;
                if $y$ is an accepting configuration: return Some($y$);
                if $y = x$: { $\mathit{count}$ += 1; break; }
            else: break;
    if $\mathit{count} = C$: return None;
    else: Halt and reject;
\end{lstlisting}
This function also requires extra $O(c \cdot k)$-bound working tapes.

Now we can build $\immerman{M} \in \nlog(O(c \cdot k), k + \partial)$ as a log-space NTM that simulates the function \texttt{judge}
and then accepts inputs if the result of \texttt{judge} is \texttt{None}.
\end{proof}

Repeatedly applying Immerman's construction collapses nested oracle machines to machines without oracles~\cite[Corollary~2]{Immerman:1988}.
\begin{theorem}\label{thm:collapse}
Let $M \in \oraclem^n(c, k)$ be a log-space $n$-nested-oracles NTM.
There exists a log-space NTM $N \in \nlog(O(c \cdot k^n), O(k \cdot n))$ such that $L(M) = L(N)$.
\end{theorem}
\begin{proof}
We eliminate oracle transitions from the innermost to the outermost for $M$ as follows.
We replace $p \xrightarrow{\in N} q$ with $N \in \oraclem^0$ with multiple transition rules that perform:
(1) save the head position $H$ to an extra tape; (2) run $N$; (3) if we reach an accepting configuration $\tuple{q_f, \, \mathcal{T}}$, then we continue $\tuple{q, H, \mathcal{T}}$.
Similarly, we replace $p \xrightarrow{\notin N} q$ with $N \in \oraclem^0$ with multiple transition rules using $\immerman{N}$ obtained by Immerman's construction.
We emphasize that each generation of $\immerman{N}$ increases $c$ and $k$ in the order of the statement of Lemma~\ref{lemma:immermans-construction}.
\end{proof}

\subsection{Membership Problem of Log-space Nested-oracles NTMs}

We now show that the membership problem of log-space nested oracle machines is in \textbf{PSPACE}.
It is a decision problem of the following form:
\begin{description}
\item[Input1] a machine $M \in \oraclem^\omega(c, k)$ ($c$ and $k$ are binary encoded).
\item[Input2] a word $w \in \Sigma^*$.
\item[Output] If $w \in_? L(M)$, return Yes. Otherwise, No.
\end{description}

To show the problem belongs to \textbf{PSPACE}, we would like to repeatedly using Theorem~\ref{thm:collapse}.
However, it is not feasible because such repeatedly application results in a log-space $O(c \cdot k^{|M|})$-bounded tapes machine in general.
It demands $O(c \cdot k^{|M|}) \cdot \log |w|$ space; thus, we cannot simulate it in polynomial size for the input sizes $|M|$ and $|w|$.
To address this problem, we adopt Immerman's construction for $\oraclem^{\omega}$ in an interpreter style.

\begin{theorem}[Membership problem of $\oraclem^{\omega}$ belongs to $\pspace$]
\label{thm:membership-oracle-machines}
Let $w$ be an input word and $M \in \oraclem^\omega(c, k)$ be an input machine where $c$ and $k$ are binary encoded.
We can decide if $w \in_? L(M)$ in polynomial space for $c$, $k$, $|w|$, and $|M|$.
\end{theorem}
\begin{proof}
First, we extend the function \texttt{one-step} for oracle transitions as follows:
\begin{lstlisting}[basicstyle=\ttfamily,mathescape,columns=spaceflexible]
def one-step$_{M_i}$($x$, $x'$):
    foreach $\delta \in \Delta(M_i)$:
        if $\delta$ is a normal transition & $c \overset{\smash{\delta}}{\Rightarrow} d$: return True;
        else if $x = (\_, i, \mathcal{T})$:
            if $\delta = p \xrightarrow{\smash{\in N}} q$:
                match $\text{judge}_{N}(  \tuple{q_{\text{init}}(N), i, \mathcal{T}} )$ with // $q_{\text{init}}(N)$$\text{\small\ is \textit{N}'s initial state}$
                | Some($\mathcal{U}$) -> return $(x' =_? \tuple{q, i, \mathcal{U}})$;
                | None -> return False;
            if $\delta = p \xrightarrow{\smash{\notin N}} r$:
                match $\text{judge}_{N}(  \tuple{q_{\text{init}}(N), i, \mathcal{T}} )$ with
                | None -> return $(x' =_? \tuple{r, i, \mathcal{T}})$;
                | Some(_) -> return False;
    return False;
\end{lstlisting}

Next, we generate the codes of $\texttt{one-step}_{M_i}$, $\texttt{counting}_{M_i}$, and $\texttt{judge}_{M_i}$ for
all oracle machines $M_i$ that appears in $M$.
Such generation is carried out in polynomial-time for $|M|$. The total size of generated code is also polynomial in $|M|$.

We can also provide an interpreter for the generated code in polynomial time for $|M|$.
While this interpreter needs a call stack for function calls, its depth is bounded by $\mathit{depth}(M) \leq |M|$.
Additionally, the size of each stack frame is bounded by $O((c \cdot k) \log|w|)$ by Proposition~\ref{prop:space-for-configurations}.
From the above argument, we can check $w \in_? L(M)$ using (nondeterministic) polynomial space with respect to $c$, $k$, $|w|$, and $|M|$.
\end{proof}

\textbf{Remark:} We can refine the above statement as follows: we decide $w \in_? L(M)$ in polynomial space for $c$, $|w|$, and $|M|$.
This is because $k \leq |M|$ is always true, even when $k$ is binary encoded.
(Of course, we assume all tapes $T_1, \ldots, T_k$ are used in some transition rules.)
However, we cannot refine it for $c$: since, if the input machine $M$ is binary encoded, $c = \Omega(2^{|M|})$ holds.
Fortunately, as we will see below, we can assume $c = 1$ when simulating \rewbl{}.
It is crucial for showing the \pspace-completeness of the membership problem of \rewbl.

\section{From \rewbl{} to Log-Space Nested Oracles NTM}\label{sec:rewbl-to-nlog}

We finally show $\rewbllang \subseteq \nlog$ by translating {\rewbl}s to $\oraclem^\omega$.

\begin{theorem}\label{rewbl-to-oracles}
Given a \rewbl{} expression $E$, we can translate it to $M \in \oraclem^{\omega}(1, O(|E|))$ such that $L(E) = L(M)$.
\end{theorem}
\begin{proof}
We inductively translate a given \rewbl{} $E$ to a $\oraclem^\omega$ $T(E)$.

As we will see below, the generated $T(E)$ has a unique source state $s$, which does not have incoming edges, and
a unique sink state $t$,  which does not have outgoing edges.
Please recall that $E$ accepts an input $w$ if it consumes all the input, i.e., if we have $\tuple{p, \Lambda} \in \sem{\tuple{E, w, 0, \iota}}$ with $p = |w|$.
So, for $M$, we add a checking transition rule to the states $s$ and $t$ of $T(E)$ as follows: $M=$\tikz[baseline=-3pt]{
\tikzset{largen/.style={fill=white, draw=black, rectangle, rounded corners=4pt, inner sep=6pt}}

\node[largen] (T1) {$T(E)$};
\node[state, fill=white, draw=black, minimum size=2pt, inner sep=2pt]
(s) [left = -.1cm of T1] {$s$};
\node[state, fill=white, draw=black, minimum size=2pt, inner sep=2pt]
(t) [right = -.1cm of T1] {$t$};

\node[state, initial, initial text={}, fill=white, draw=black, minimum size=5pt, inner sep=2pt] (init) [left = .7cm of s] {$i$};
\draw[->] (init) edge node[above]{\small $\vdash\,\mid+1$} (s);

\node[state, accepting, fill=white, draw=black, minimum size=5pt, inner sep=2pt] (u) [right = .7cm of t] {$f$};
\draw[->] (t) edge node[above]{\small $\dashv\ \mid\,0$} (u);
} where $i$ and $f$ are initial and accepting states of $M$.

\paragraph*{First Step: REGEX to Log-Space NTM}

First, we give a translation for the REGEX part of REWB as follows:
\[
\sigma \Mapsto 
\tikz[baseline=0pt]{
\tikzset{smalln/.style={state, minimum size=0pt, inner sep=1pt, fill=white, draw=black}}

\node[smalln] (s) {$s$};
\node[smalln] (t) [right = 1.4cm of s] {$t$};

\draw[->] (s) edge node[above] {\small $\sigma \mid +1$} (t);
},
\qquad
E_1 E_2 \Mapsto
\tikz[baseline=0pt]{
\tikzset{largen/.style={fill=white, draw=black, rectangle, rounded corners=4pt, inner sep=6pt}}

\node[largen] (T1) {$T(E_1)$};
\node[state, fill=white, draw=black, minimum size=0pt, inner sep=0pt]
(s1) [left = -.1cm of T1] {$s_1$};
\node[state, fill=white, draw=black, minimum size=0pt, inner sep=0pt]
(t1) [right = -.1cm of T1] {$t_1$};

\node[largen] (T2) [right = .8cm of t1] {$T(e_2)$};
\node[state, fill=white, draw=black, minimum size=0pt, inner sep=0pt]
(s2) [left = -.1cm of T2] {$s_2$};
\node[state, fill=white, draw=black, minimum size=0pt, inner sep=0pt]
(t2) [right = -.1cm of T2] {$t_2$};

\draw[->] (t1) edge[above] node{\small nop} (s2);
}\ ,
\]
\[
(E)^* \Mapsto \tikz[baseline=0pt]{
\tikzset{largen/.style={fill=white, draw=black, rectangle, rounded corners=4pt, inner sep=6pt}}

\node[largen] (T1) {$T(E)$};
\node[state, fill=white, draw=black, minimum size=0pt, inner sep=.2pt]
(s1) [left = -.1cm of T1] {$\,s'$};
\node[state, fill=white, draw=black, minimum size=0pt, inner sep=.2pt]
(t1) [right = -.1cm of T1] {$\,t'$};

\node[state, fill=white, draw=black, minimum size=0pt, inner sep=.5pt]
 (u) [left = .5cm of s1] {$\,s\,$};
\draw[->] (u) edge[above] node{\small nop} (s1);
\node[state, fill=white, draw=black, minimum size=0pt, inner sep=.5pt]
 (v) [right = .5cm of t1] {$\,t\,$};
 \draw[->] (t1) edge[above] node{\small nop} (v);
 
\draw[->, thick]  (t1) -- ++(0, .7cm) -| (s1) node[above,pos=0.25] {nop};
}\ ,
\qquad
E_1 + E_2 \Mapsto
 \tikz[baseline=0pt]{
 \tikzset{largen/.style={fill=white, draw=black, rectangle, rounded corners=4pt, inner sep=6pt}}
\tikzset{smalln/.style={state, minimum size=0pt, inner sep=1pt, fill=white, draw=black}}

 \node[smalln] (s) {$s$};
 \node[smalln] (t) [right = 2.5cm of s] {$t$};
 
\node[largen] (T1) [above = 0cm of s, xshift=1.4cm] {$T(E_1)$};
\node[state, fill=white, draw=black, minimum size=0pt, inner sep=0pt]
(s1) [left = -.1cm of T1] {$s_1$};
\node[state, fill=white, draw=black, minimum size=0pt, inner sep=0pt]
(t1) [right = -.1cm of T1] {$t_1$};

\node[largen] (T2) [below = 0cm of s, xshift=1.4cm] {$T(E_2)$};
\node[state, fill=white, draw=black, minimum size=0pt, inner sep=0pt]
(s2) [left = -.1cm of T2] {$s_2$};
\node[state, fill=white, draw=black, minimum size=0pt, inner sep=0pt]
(t2) [right = -.1cm of T2] {$t_2$};
\draw[->] (s) edge[above] node[xshift=-10pt]{\small nop} (s1);
\draw[->] (s) edge[above] node[below, xshift=-10pt]{\small nop} (s2);
\draw[->] (t1) edge[auto] node[auto]{\small nop} (t);
\draw[->] (t2) edge[auto] node[below, xshift=10pt]{\small nop} (t);
}\ ,
\]
where the edges labelled with ``nop'' are the same ones used in Example~\ref{log-space-example1}, which just change states.

This translation is identical to the McNaughton--Yamada--Thompson algorithm, which is well-known and found in textbooks of automata theory.
Now, $L(E) = L(M)$ clearly holds.

\paragraph*{Second Step: Variable renaming and Translating $(E)_x$ and $\var{x}$}

Before translating $(E)_x$ and $\var{x}$, to simplify our translation,
we perform \emph{variable renaming} like alpha-conversion in lambda calculus to remove patterns such that $(\cdots x \cdots)_x$
where a variable $x$ appears inside an expression capturing $x$.
Formally, if a variable $x$ appears in $E$ of a capturing expression $(E)_x$,
then we change $(E)_x$ to $?((E)_y) (\var{y})_x$ using a fresh variable $y$ without modifying $E$.

We perform this variable renaming from the innermost to the outermost. For example,
\begin{itemize}
\item an expression $(aa)_x (\var{x} \var{x})_x$ is translated to $(aa)_x ?((\var{x} \var{x})_y) (\var{y})_x$.
\item an expression $(a)_x (b)_y (?[(\var{x} \var{x})_x] aa \var{y})_x$ is first translated to
$(a)_x (b)_y (?[?((\var{x} \var{x})_\alpha) (\var{\alpha})_x] aa \var{y})_x$ and then to
$(a)_x (b)_y ?((?[?((\var{x} \var{x})_\alpha) (\var{\alpha})_x] aa \var{y})_\beta) (\var{\beta})_x$.
\end{itemize}

After variable renaming, we focus on the part $(E)_x$ of REWB:
\[
(E)_x \Mapsto
\tikz[baseline=0pt]{
\tikzset{largen/.style={fill=white, draw=black, rectangle, rounded corners=4pt, inner sep=6pt}}
\tikzset{smalln/.style={state, minimum size=0pt, inner sep=1pt, fill=white, draw=black}}

\node[largen] (T1) {$T(E)$};
\node[smalln] (s1) [left = -.1cm of T1] {$\,s'$};
\node[smalln] (t1) [right = -.1cm of T1] {$\,t'$};

\node[smalln] (u) [left = 4.4cm of s1] {$\,s\,$};
\draw[middle dotted line] (u) -- node[above]{\small \begin{tabular}{c}save the current head position \\ in binary form to the tape $T_{x_l}$\end{tabular}} (s1);

\node[smalln] (v) [right = 4.4cm of t1] {$\,t\,$};
\draw[middle dotted line] (t1) -- node[above]{\small \begin{tabular}{c}save the current head position \\ in binary form to the tape $T_{x_r}$\end{tabular}} (v);
}\,.
\]
In order to keep the start position of the variable $x$,
we first copy the current head position to the special working tape $T_{x_l}$ in binary form.
Then, we execute the expression $E$ by running from the state $s$ to $t$.
Finally, we record the new head position into the working tape $T_{x_r}$.
Now, $w[x_l\,..\,x_r) = w[x_l] w[x_l + 1] \cdots w[x_r -1]$ is a substring matched with the expression $E$
where $x_l$ and $x_r$ are the numbers corresponding to the contents of $T_{x_l}$ and $T_{x_r}$.

Next, we focus on the part $\var{x}$ of REWB:
\[
\var{x} \Mapsto
\tikz[baseline=0pt]{
\tikzset{smalln/.style={state, minimum size=0pt, inner sep=1pt, fill=white, draw=black}}

\node[smalln] (s)  {$\,s\,$};
\node[smalln] (t) [right = 9cm of s] {$\,t\,$};

 \draw[middle dotted line] (s) -- node[above]{\small check if the current head starts with the substring of $w[x_l\,..\,x_r)$} (t);
}\,.
\]
As we have seen above, the substring $w[x_l\,..\,x_r)$ denotes the value of the variable $x$.
This checking task is accomplished using an extra tape $T_\text{tmp}$ without changing $T_{x_l}$ and $T_{x_r}$.

On these translations, it holds that $L(E) = L(M)$ for any REWB expressions $E$.

\paragraph*{Third Step: \rewbl{} to $\oraclem^{\omega}$}\label{subsec:rewbl-to-logntm}

For our construction, we need that lookahead are augmented with a continuation $K$, for instance $?(E)K$ and $!(E) K$.
If a given expression does not satisfy this property, we add the expression $\epsilon$, which always succeeds.
For example, $(E_1 E_2 !(E_3))^* E_4 \Mapsto (E_1 E_2 (!(E_3) \epsilon))^* E_4$.
Now, we can easily translate $?(E)K$ and $!(E)K$ using oracle transition rules as follows:
\[
?(E)K \Mapsto
\tikz[baseline=0pt]{
\tikzset{largen/.style={fill=white, draw=black, rectangle, rounded corners=4pt, inner sep=6pt}}
\tikzset{smalln/.style={state, minimum size=0pt, inner sep=1pt, fill=white, draw=black}}

\node[largen] (T1) {$T(K)$};
\node[smalln] (s1) [left = -.1cm of T1] {$\,s'$};
\node[smalln] (t1) [right = -.1cm of T1] {$\,t'$};

\node[smalln] (u) [left = 2 cm of s1] {$\,s\,$};
\draw[->] (u) -- node[above]{\small $\in T(E)$} (s1);
},\quad
!(E)K \Mapsto
\tikz[baseline=0pt]{
\tikzset{largen/.style={fill=white, draw=black, rectangle, rounded corners=4pt, inner sep=6pt}}
\tikzset{smalln/.style={state, minimum size=0pt, inner sep=1pt, fill=white, draw=black}}

\node[largen] (T1) {$T(K)$};
\node[smalln] (s1) [left = -.1cm of T1] {$\,s'$};
\node[smalln] (t1) [right = -.1cm of T1] {$\,t'$};

\node[smalln] (u) [left = 2 cm of s1] {$\,s\,$};
\draw[->] (u) -- node[above]{\small $\notin T(E)$} (s1);
}.
\]

By the semantics of our nested oracle machines, it is clear that  $L(E) = L(M)$.
\end{proof}

Now, Theorem~\ref{thm:simulate-2mfa}, \ref{thm:collapse}, and~\ref{rewbl-to-oracles}, and Theorem~\ref{thm:tqbf-and-rewbl}~and~\ref{thm:membership-oracle-machines} imply our main results.
\begin{corollary}
$\rewbllang = \nlog$.
\end{corollary}

\begin{corollary}
The membership problem of \rewbl{} is \textbf{PSPACE}-complete.
\end{corollary}

\section{Future Work and Conclusion}\label{sec:conclusion}

We have shown that the language class \rewbl{}---regular expressions plus backreferences and lookaheads (without any restrictions)---captures the class \nlog.
Our result closes the expressiveness about \rewbl{}.
Furthermore, we have shown that the membership problem of \rewbl{} is \pspace-complete.
On the other hand, it remains unclear whether \prewb{} and \nrewb{} are proper subclasses of \rewbl{}.
For example, we conjecture that \prewb{} cannot recognize the language $L_{\text{prime}} = \set{ a^ n : \text{$n$ is a prime number} }$.\footnote{
\rewbl{} recognizes $L_{\text{prime}}$ as follows:
(1) we define $E_{\text{composite}} = (aa\,a^*)_w \var{w} \var{w}^* \$$, which recognizes composite numbers;
(2) then, we define $E_{\text{non-prime}} = a\$ + E_{\text{composite}}$, which recognizes non-prime numbers;
(3) finally, we define $E_{\text{prime}} = !(E_{\text{non-prime}})\ a^*\ \$$.
Readers interested in conducting a prime test with $E_{\text{prime}}$ can try the following expression in C\#:
$\texttt{@"\textasciicircum(?!((a\$)|((?<W>(aa+))(\textbackslash k<W>+)\$))).*\$"}$.
}
It is known that the language can be represented by \nrewb{}.
Also, we conjecture that some \nlog-languages cannot be recognized by \nrewb{}.

\section*{Acknowledgement}
I thank anonymous reviewers for their detailed comments on a previous version of this paper.
I also sincerely thank anonymous reviewers for their careful reading and invaluable comments on the current version.
All comments helped me to significantly improve the presentation of this paper and the clarity of proofs and constructions.
This work was supported by JST, CREST Grant Number JPMJCR21M3.

\bibliographystyle{plainurl}
\bibliography{mybib}

\appendix

\section{Proof of Theorem~\ref{thm:unary-non-indexed}: Non Indexed Languages Accepted by \prewb{} and \nrewb{}}

We said that the language $L(E)$ grows faster than the language $L_{\mathit{2exp}} = \set{ a^{2^{2^n}} : n \in \mathbb{N} }$ in Section~\ref{subsec:unary-non-indexed-languages}.
To complete the proof sketch for the \nrewb{} part, we introduce the notion of faster growing.
To formally state the notion, we first introduce the (growth) order function $\mathcal{G}_L : \mathbb{N} \to \mathbb{N}$ for unary languages $L$ as follows:
\[
\mathcal{G}_L(n) = \text{the length of $n$-th shortest word in $L$}.
\]
For example,
$\mathcal{G}_{L_{\mathit{1exp}}}(n) = 2^n$ where $L_{\mathit{1exp}} = \set{ a^{2^n} : n \in \mathbb{N} }$,
$\mathcal{G}_{L_{\mathit{2exp}}}(n) = 2^{2^n}$, and
$\mathcal{G}_{L_{\text{2tower}}}(n) =  2\!\underbrace{{}^{2^{\rddots2}}}_n$.
Using $\mathcal{G}$, for two unary languages $L_1$ and $L_2$, we say $L_1$ \emph{grows faster} than $L_2$ if $\mathcal{G}_{L_1}(n) = \Omega(\mathcal{G}_{L_2}(n))$.

We recall an important result about indexed languages: i.e., for any indexed language $I$, its growth order is $O(2^n)$. In other words, $\mathcal{G}_I(n) = O(2^n)$.
This result was shown by pumping or shrinking arguments on indexed languages~\cite{Hayashi:1973,Gilman:1996}.
Since the growth rate of $L(E'_{\mathit{2exp}})$ surpasses the order $2^n$ clearly, \nrewb{} can represent a non-indexed language.
\begin{lemma}\label{lem:fast-growing-neglang}
\nrewb{} recognizes a language that grows faster than $L_{\text{2tower}}$.
\end{lemma}

Next, let us consider the \prewb{} part.
Combining the construction the example of Section~\ref{subsec:unary-non-indexed-languages} and a technique, which we call \emph{halving},
we can give a \prewb{} expression representing $L_{\mathit{2exp}}$.
\begin{lemma}
\prewb{} can recognize the doubly exponential language $L_{\mathit{2exp}} = \set{ a^{2^{2^n}} : n \in \mathbb{N} }$.
\end{lemma}
\begin{proof}
We proceed as the example of Section~\ref{subsec:unary-non-indexed-languages}; therefore,
\begin{itemize}
\item By $E_1 = ?(a)_m ?((\var{n}a)_n (\var{m}\var{m})_m)^*$, we make $n$ and $m = 2^n$ nondeterministically.
\item Also, by $E_2 = ?(a)_y ?((\var{x}a)_x (\var{y}\var{y})_y)^*$, we make $x$ and $y = 2^x$ nondeterministically.
\end{itemize}
Then, we need to check $x =m$; however, due to the absence of negative lookaheads, especially $\$$,
we cannot do $?(\Sigma^* ?(\var{m} \$) ?(\var{x} \$))$ to check it.

Instead of using $\$$, we use another technique, which heavily depends on the acceptance condition of \rewbl{}---we need to consume all the inputs.

As idea, we build $m_{1/2}$ (resp. $x_{1/2}$) that contains the half $a$'s of $m$ (resp. $x$).

Then, $E = E_1 E_2 \var{y} ?(\var{m}) ?(\var{x}) \var{m_{1/2}} \var{x_{1/2}}$ accepts $w$ iff $m = x$;
therefore, $L(E) = \set{ a^{2^{2^n}} a^{2^n} : n \in \mathbb{N} }$.

To show that $E$ accepts $w$ iff $m = x$,
we show (1) $m = x \implies w \in L(E)$; and (2) $m \neq x \implies w \notin L(E)$.

In the case $m = x$, it suffices to showing $\var{m} = (\var{m_{1/2}} \var{x_{1/2}})$. It is clear.

In the case $m < x$, $(\var{m_{1/2}} \var{x_{1/2}}) < \var{x}$.

Therefore, if we reach the point $\var{m_{1/2}} \var{x_{1/2}}$, we cannot consume the rest part since $\var{y} \var{x}$ is a prefix of $w$
($\var{y} \var{m_{1/2}} \var{x_{1/2}} \preceq \var{y} \var{x} \preceq w$).

In the case $m > x$, we cannot consume all the input $w$; hence, $w \notin L(E)$.
\end{proof}

Combining these lemmas, we obtain the following theorem.
\begin{customthm}{3}
\prewb{} and \nrewb{} can represent non indexed languages.
\end{customthm}

\section{Proof of Theorem~\ref{thm:emptiness-of-unary}: Undecidability of Emptiness Problems of Unary \prewb{} and Unary \nrewb{}}

Here we consider unary \prewb{} and unary \nrewb{} whose input alphabet is a single set $\Sigma = \set{a}$.

\begin{lemma}
The emptiness problem of unary \nrewb{} is undecidable.
\end{lemma}
\begin{proof}
We use the undecidability of checking if a given Diophantine equation has a solution of \emph{natural} numbers~\cite{Matiyasevich:1993}.

To illustrate our idea, let us consider a Diophantine equation and transform it to one where each coefficient is positive as follows:
\[
D: 2 x^3 - (x-1)y = -2 \quad \Rightarrow \quad 2 x^3 + y + 2 = x y.
\]
This has a solution e.g., $(x, y) = (2, 18)$.

We nondeterministically build $x$, $y$, $x^2$, $x^3$, and $xy$ in this order in inputs using the expression
\[
E_\text{GEN} =
(a^*)_x (a^*)_y
((\var{x^2}\var{x})_{x^2} (\var{w_x} a)_{w_x})^*
((\var{x^3}\var{x^2})_{x^3} (\var{w'_x} a)_{w'_x})^*
((\var{xy}\var{y})_{xy} (\var{w''_x} a)_{w''_x})^*
\]
as follows:
\[
\underbrace{\cdots}_{x}
\underbrace{\cdots}_{y}
\mid
\var{x}\,a~~~2\var{x}\ a^2\,\cdots\,
\underbrace{n\var{x}}_{{x^2}}
\underbrace{a^n}_{w_{x}}
\mid
\var{{x^2}}\,a
%~~~2\var{{x^2}}\ a^2
\,\cdots\,
\underbrace{m\var{{x^2}}}_{{x^3}}
\underbrace{a^m}_{w'_{x}}
\mid \\
\var{y}\,{a}
%~~~2\var{y}\ a^2
\,\cdots\,
\underbrace{l \var{y}}_{{xy}}
\underbrace{a^l}_{w''_{x}}
\]

If we can ensure $\var{x} = \var{w_x} = \var{w'_x} = \var{w''_x}$,
then $\var{x^2} = a^{x^2}$, $\var{x^3} = a^{x^3}$, and $\var{xy} = a^{x y}$ are derived.
Then, it suffices to performing $!!(\var{x^3} \var{x^3} \var{y} a a \$) !!(\var{xy} \$)$ to check $2 x^3 + y + 2 = xy$.
Here $!!(E)$ is a shortened version of $!(!(E))$, which means quasi positive lookaheads where we cannot update variables in such lookaheads.

To ensure $\var{x} = \var{w_x} = \var{w'_x} = \var{w''_x}$, we extend the above expression as follows:
\[
E_\text{check} = !!(\var{x^3} \var{x^3} \var{y} a a \$) !!(\var{w_{xy}} \$)\ a^*\ !!(\var{v_x} \$) !!(\var{w_x} \$) !!(\var{w'_x} \$) !!(\var{w''_x} \$)\ a^*.
\]
The entire expression is $E = E_\text{GEN} E_\text{check}$, and then $L(E) \neq \emptyset$ iff $D$ has a solution.
\end{proof}

\def\half{{}^1{\mskip -5mu/\mskip -3mu}_2}

Next, we focus on unary \prewb{}.
\begin{lemma}\label{thm:PCP-REWBpos}
The emptiness problem of unary \prewb{} is undecidable.
\end{lemma}
\begin{proof}
Let us consider the following (binary) PCP instance:
\[
\begin{array}{c|c@{\quad}c@{\quad}c}
 & R_1 & R_2 & R_3 \\\hline
\alpha & a & ab & bba \\
\beta & baa & aa & bb
\end{array} \Longrightarrow
\begin{array}{c|c@{\quad}c@{\quad}c}
& R_1 & R_2 & R_3 \\\hline
\alpha & 2 & 24 & 442 \\
\beta & 422 & 22 & 44
\end{array} \qquad
\begin{tabular}{c}
\small We replace the characters by 2 and 4, \\
\small \ \ here $a \mapsto 2$ and $b \mapsto 4$. \\
\small It has a solution $R_3 R_2 R_3 R_1$; \\ 
\small $442\,24\,442\,2 = 44\,22\,44\,422$.
\end{tabular}
\]
Starting from the number 2, we can simulate applying each rule.
For example, apply $R_3 = 442$ of $\alpha$, we first multiply the current value by $10^{|442|=3}$ and then add $442$.
Repeatedly applying this, we obtain
\[
\alpha: 2 \xrightarrow{R_3} 2 \cdot  {10}^3 + 442 \xrightarrow{R_2} 2442 * {10}^2 + 24
\xrightarrow{R_3} 244224442 \xrightarrow{R_1} 2442244422,
\]
where we drop the leftmost 2, then we obtain the correct value 442244422.
To build an expression $E$ that satisfies $L(E) \neq \emptyset$ iff there is a solution of the PCP instance,
we take the following idea: 
\begin{enumerate}
\item Besides calculating $\alpha$ and $\beta$, we also compute their halves $\alpha_\half$ and $\beta_\half$.
\item If $\alpha = \beta$, $\alpha_\half + \beta_\half = \alpha = \beta$. Otherwise, $\alpha_\half + \beta_\half < \max{\alpha, \beta}$.
\item Therefore, it holds that $?(\alpha) ?(\beta) (\alpha_\half \beta_\half)$ matches $a^n$ iff the PCP instance has a solution, whose result is $a^n$. 
\item Please recall that \rewbl{} determines acceptance based on if it can consume the entire input.
\end{enumerate}

To initialize variables, we consider $E_{\text{init}} = ?(aa)_\alpha ?(a)_{\alpha_{\half}} ?(aa)_\beta ?(a)_{\beta_{\half}}$.
We use $E_{R_1} = \left(
\begin{array}{c}
?((10 \times \var{\alpha}) a^{2})_{\alpha}
?((10 \times \var{\alpha_{\half}}) a^{1})_{\alpha_{\half}} \\
?((10^3 \times \var{\beta}) a^{422})_{\beta}
?((10^3 \times \var{\beta_{\half}}) a^{211})_{\beta_{\half}}
\end{array}
\right)$ to reflect the rule $R_1$.
Also defining $E_{R_2}$ and $E_{R_3}$, then we just consider $E = E_\text{init} (E_{R_1} + E_{R_2} + E_{R_3})^*\ ?(\var{\alpha})\ ?(\var{\beta}) \var{\alpha_{\half}}\,\var{\beta_{\half}}$. Now, $L(E) \neq \emptyset$ iff the PCP has a solution.
\end{proof}

Combining these lemmas, we obtain the following theorem.
\begin{customthm}{4}
The emptiness problems of \prewb{} over a unary alphabet and \nrewb{} over a unary alphabet are undecidable.
\end{customthm}

\end{document}